\documentclass[10pt,reqno]{amsart}
     \makeatletter
     \def\section{\@startsection{section}{1}%
     \z@{.7\linespacing\@plus\linespacing}{.5\linespacing}%
     {\bfseries
     \centering
     }}
     \def\@secnumfont{\bfseries}
     \makeatother
\usepackage{amsmath}
\usepackage{amssymb}
\usepackage[latin1]{inputenc}
\usepackage{bbm}
\newcommand*{\cA}{\mathcal{A}}
\newcommand*{\cB}{\mathcal{B}}
\newcommand*{\cC}{\mathcal{C}}
\newcommand*{\C}{\mathbb{C}}

\newcommand*{\N}{\mathbb{N}}
\newcommand*{\1}{\mathbbm{1}}
\newcommand*{\R}{\mathbb{R}}
\newcommand*{\halb}{\frac{1}{2}}
\newcommand{\abs}[1]{\left|#1\right|}
\newcommand{\norm}[1]{\left\|#1\right\|}

\renewcommand{\and}{\text{ and }}

\newtheorem{theorem}{Theorem}[section]
\newtheorem{lemma}[theorem]{Lemma}
\newtheorem{proposition}[theorem]{Proposition}
\newtheorem{corollary}[theorem]{Corollary}
\theoremstyle{definition}
\newtheorem{definition}[theorem]{Definition}
\newtheorem{example}[theorem]{Example}
\theoremstyle{remark}
\newtheorem{remark}[theorem]{Remark}
\numberwithin{equation}{section}
\begin{document}

\title[Charged Particle as Hida Distribution]{The Hamiltonian Path Integrand for the Charged Particle in a Constant Magnetic field as White Noise Distribution}

\author[Wolfgang Bock]{Wolfgang Bock}
\address{Functional Analysis and Stochastic Analysis Group, \\
Department of Mathematics, \\
University of Kaiserslautern, 67653 Kaiserslautern, Germany}
\email{bock@mathematik.uni-kl.de}
\urladdr{http://www.mathematik.uni-kl.de/$\sim$bock}

\author[Martin Grothaus]{Martin Grothaus }
\address{Functional Analysis and Stochastic Analysis Group, \\
Department of Mathematics, \\
University of Kaiserslautern, 67653 Kaiserslautern, Germany}
\email{grothaus@mathematik.uni-kl.de}
\urladdr{http://www.mathematik.uni-kl.de/$\sim$grothaus}



\begin{abstract}
The concepts of Hamiltonian Feynman integrals in white noise analysis are used to realize as the first velocity dependent potential the Hamiltonian Feynman integrand for a charged particle in a constant magnetic field in coordinate space as a Hida distribution. For this purpose the velocity dependent potential gives rise to a generalized Gauss kernel. Besides the propagators also the generating functionals are obtained. 
\end{abstract}

\maketitle

\section{Introduction}
The Feynman path integral is a very successfully applied object. Although the first aim of Feynman was to develop path integrals based on a Lagrangian, they also can be used for various systems which have a law of least action, see e.g.\cite{F48}.\\
Since classical quantum mechanics is based on a Hamiltonian formulation rather than a Lagrangian one, it is worthwhile to take a closer look to the so-called Hamiltonian path integral, which means the Feynman integral in phase space.
Feynman gave a heuristic formulation of the phase space Feynman Integral
\begin{equation}\label{psfey}
K(t,y|0,y_0)= {\rm N} \int_{x(0)=y_0, x(t)=y} \int \exp\left(\frac{i}{\hbar} S(x)\right) \prod_{0<\tau<t}  \frac{dp(\tau)}{(2\pi)^d} dx(\tau)
\end{equation}
in \cite{Fe51}.
Here the action (and hence the dynamic) is expressed by a canonical (Hamiltonian) system of generalized space variables $x$ and their corresponding conjugate momenta $p$. The canonical variables can be found by a Legendre-transformation, see e.g.~\cite{Sch07}.
The Hamiltonian action:
\begin{equation*}
S(x,p,t)=\int_0^t p(\tau)\dot{x}(\tau) -H(x(\tau),p(\tau),\tau) d\tau,
\end{equation*}
where 
\begin{equation*}
H(x,p,t)=\frac{1}{2m}p^2 +V(x,p,t)
\end{equation*}
is the Hamilton function and given by the sum of the kinetic energy and the potential.
Note that both integrals, the Feynman integral as well as the Hamiltonian path integral are thought to be integrals w.r.t~a flat, i.e.~translation invariant measure on the infinite dimensional path space. Such a measure does not exist, hence the integral in \ref{psfey} is not a mathematical rigorous object. Nevertheless there is no doubt that it has a physical meaning.\\
The Hamiltonian setting has many advantages such as e.g.~:
\begin{itemize}
\item the semi-classical limit of quantum mechanics is more natural in an Hamiltonian setting, i.e. the phase space is more natural in classical mechanics than the configuration space, see also \cite{AHKM08, KD82} and the references therein.
\item in \cite{DK83} the authors state, that potentials which are time-dependent or velocity dependent should be treated with the Hamiltonian path integral. 
\item momentum space propagators can be investigated.
\end{itemize}
There are many attempts to give a meaning to the Hamiltonian path integral as a mathematical rigorous object. Among these are analytic continuation of probabilistic integrals via coherent states \cite{KD82, KD84} and infinite dimensional distributions e.g.~\cite{DMN77}. Another approach by Albeverio et al.~uses Fresnel integrals e.g.~\cite{AHKM08, AGM02} and most recently a method using time-slicing was developed by Naoto Kumano-Go \cite{Ku11}. 
As a guide to the literature on many attempts to formulate these ideas we point out the list in \cite{AHKM08}.\\
Here we choose a white noise approach. White noise analysis is a mathematical framework which offers generalizations of concepts from finite-dimensional analysis, like differential operators and Fourier transform to an infinite-dimensional setting. We give a brief introduction to white noise analysis in Section 2, for more details see \cite{Hi80,BK95,HKPS93,Ob94,Kuo96}. Of special importance in white noise analysis are spaces of generalized functions and their characterizations. In this article we choose the space of Hida distributions, see Section 2.\\
The idea of realizing Feynman integrals within the white noise framework goes back to \cite{HS83}. As ansatz for the Feynman integrand in configuration space one has
\begin{multline}\label{integrandpot}
I_V = {\rm Nexp}\left( \frac{i}{2\hbar}\int_0^t \dot{x}^2(\tau)d\tau +\frac{1}{2}\int_0^t \dot{x}(\tau)^2 d\tau\right)\\
\times \exp\left(-\frac{i}{\hbar} \int_0^t V(x(\tau),\dot{x}(\tau)) \, d\tau\right) \cdot \delta(x(t)-y),
\end{multline}
where $x(t)=y_0+ B_t$ is a Brownian path starting in $y_0$. In equation \eqref{integrandpot} the first integral in the exponential represents the kinetic energy and the second integral the compensation of the Gaussian fall-off. The delta function (Donsker's Delta) pins the paths at the end time point in the end point. The normalized exponential as well as the delta function exists as well-defined objects in the space of Hida distributions.
With this concept many authors constructed the Feynman integrand for a large class of Lagrange functions and thus the Feynman integral as generalized expectation of the integrand w.r.t.~ the Gaussian measure see e.~g.~\cite{BCB02}, \cite{FPS91}, \cite{FOS05},
\cite{GKSS97}, \cite{HKPS93}, \cite{KS92}, \cite{KWS97},
\cite{L06}, \cite{SS04}, \cite{BGJ13}.\\
In \cite{BG11} the concepts from \cite{GS98a} are used to give a mathematical rigorous realization to the Hamiltonian path integrand as a Hida distribution in the case of non-velocity dependent potentials. There also the canonical commutation relations could be shown in the sense of \cite{FeHi65}. To obtain this it is used that the white noise analysis is not only giving a meaning to the integral as a generalized expectation of the integrand but provides also the generating functional of the Green's function to the Schrödinger equation
\\
In this article we apply the methods from \cite{BG11} to the physical system of a charged particle in a constant magnetic field and thus a velocity dependent potential. Investigations of this system in configuration space have been done in \cite{G96}, \cite{J10} and \cite{BGJ13}, also using the white noise approach. The motion we consider in the plane orthogonal to the direction of the magnetic field. 
For the corresponding Hamiltonian action one finds
\begin{equation*}
S(q,p,T) = \int_0^{T} {\bf p} \dot{\vec{{\bf x}}} - \frac{1}{2m} \bigg({\bf p} -\frac{q}{c} {\bf A}({\bf x})  \bigg)^2 \; d\tau,
\end{equation*}
where $${\bf p}=\left( \begin{array}{l} p_1\\p_2\\p_3\end{array}\right)\text{ and }{\bf x}=\left( \begin{array}{l} x_1\\x_2\\x_3\end{array}\right),$$
respectively, see e.g.~\cite[form.(2.49), p.103]{Sch07}. Moreover $q$ is the charge of the particle, $c$ is the speed of light and ${\bf A}$ is a vector potential. Note that a multiplication of the vectors above is thought of a the euclidean scalar product, i.e. 
$${\bf p} {\bf x} = p_1 x_1 + p_2 x_2 +p_3 x_3.$$
Here we consider the case of a constant magnetic field along the $x_3$-axis, i.e.~the axis orthogonal to the plane spanned by $x_1$ and $x_2$. We have ${\bf B}=(0,0,B_z)$.
With the relation 
$${\bf B}= \mathrm{rot}({\bf A}),$$
we have
$$A= B_z \left(\begin{array}{c} -x_2\\x_1\\0 \end{array}\right).$$
Thus we obtain
\begin{multline} \label{cpaction}
S({\bf x},{\bf p},t) = \int_0^{t} {\bf p} \dot{{\bf x}} - \frac{1}{2m} \bigg({\bf p} -\frac{q}{c} B_z \left(\begin{array}{c} -x_2\\x_1\\0 \end{array}\right)  \bigg)^2 \; d\tau\\ 
= \int_0^{t} {\bf p} \dot{{\bf x}} - \frac{1}{2m}(p_1^2+p_2^2+p_3^2) +\frac{q}{mc} B_z (x_1 p_2 -x_2 p_1) -\frac{q^2 B_z^2}{2mc^2}(x_1^2+x_2^2) d\tau.
\end{multline}
Note that the free motion along the $x_3$-axis separates independently from the motion in the plane. Thus in the following we consider the planar motion. 
Then we propose the following formal ansatz for the Feynman integrand in Phase space for a charged particle in a constant magnetic field staring in $(0,0)$ at time $0$ and endpoint $(y_1,y_2)$ at time $t$:
\begin{multline}\label{anpsfeycp}
I_{CP} =\\
{\rm Nexp}\left( \frac{i}{\hbar}\int_0^{t} {\bf p} \dot{{{\bf x}}} - \frac{1}{2m}(p_1^2+p_2^2) d\tau +\frac{1}{2}\int_{t_0}^t \dot{x_1}(\tau)^2 +\dot{x_2}(\tau)^2  +p_1(\tau)^2+ p_2(\tau)^2 d\tau\right)\\ 
\times \exp\left(-\frac{i}{h} \int_{t_0}^t \frac{1}{2m}(p_1^2+p_2^2) -\frac{q}{mc} B_z (x_1 p_2 -x_2 p_1) -\frac{q^2 B_z^2}{2mc^2}(x_1^2+x_2^2) d\tau\right)\\  \times \delta(x_1(t)-y_1)\delta(x_2(t)-y_2).
\end{multline}
In this expression the sum of the first and the third integral is the action $S(x,p)$, and the Donsker's delta function serves to pin trajectories to $y$ at time $t$. The second integral is introduced to simulate the Lebesgue integral by compensation of the fall-off of the Gaussian measure in the time interval $({t_0},t)$. Furthermore, as in Feynman's formula we need a normalization which turns out to be infinity and will be implemented by the use of a normalized exponential as in Chapter \ref{GGK}. Note that it is a priori not clear which terms one has to normalize to obtain the right physics, i.e.~ the propagator and the commutation relations. We used in this article the same normalization procedure as in \cite{BG11}, i.e.~ we normalized the kinetic energy and the term which simulates the flat measure.\\

These are the core results of this article:
\begin{itemize}
\item The concepts of generalized Gauss kernels from \cite{GS98a} and \cite{BG11} are used to construct the Feynman integrand for a charged particle in a constant magnetic field $I_{CP}$ as in \ref{anpsfeycp} as a Hida distribution, see Theorem \ref{magnetictheorem}.
\item The results in Theorem \ref{magnetictheorem} provide us with the generating functional for a charged particle in a constant magnetic field. 
\item The generalized expectations (generating functional at zero) yields the Green's functions to the corresponding Schrödinger equation.
\end{itemize}

\section{White Noise Analysis}
\subsection{Gel'fand Triples}
Starting point is the Gel'fand triple $S_d(\R) \subset L^2_d(\R) \subset S'_d(\R)$ of the $\R^d$-valued, $d \in \N$, Schwartz test functions and tempered distributions with the Hilbert space of (equivalence classes of) $\R^d$-valued square integrable functions w.r.t.~the Lebesgue measure as central space (equipped with its canonical inner product $(\cdot, \cdot)$ and norm $\|\cdot\|$), see e.g.~ \cite[Exam.~11]{W95}.
Since $S_d(\R)$ is a nuclear space, represented as projective limit of a decreasing chain of Hilbert spaces $(H_p)_{p\in \N}$, see e.g.~\cite[Chap.~2]{RS75a} and \cite{GV68}, i.e.~
\begin{equation*}
S_d(\R) = \bigcap_{p \in \N} H_p,
\end{equation*}
we have that $S_d(\R)$ is a countably Hilbert space in the sense of Gel'fand and Vilenkin \cite{GV68}. We denote the inner product and the corresponding norm on $H_p$ by $(\cdot,\cdot)_p$ and $\|\cdot\|_p$, respectively, with the convention $H_0 = L^2_d(\R)$.
Let $H_{-p}$ be the dual space of $H_p$ and let $\langle \cdot , \cdot \rangle$ denote the dual pairing on $H_{p} \times H_{-p}$. $H_{p}$ is continuously embedded into $L^2_d(\R)$. By identifying $L_d^2(\R)$ with its dual $L_d^2(\R)'$, via the Riesz isomorphism, we obtain the chain $H_p \subset L_d^2(\R) \subset H_{-p}$.
Note that $\displaystyle S'_d(\R)= \bigcup_{p\in \N} H_{-p}$, i.e.~$S'_d(\R)$ is the inductive limit of the increasing chain of Hilbert spaces $(H_{-p})_{p\in \N}$, see  e.g.~\cite{GV68}.
We denote the dual pairing of $S_d(\R)$ and $S'_d(\R)$ also by $\langle \cdot , \cdot \rangle$. Note that its restriction on $S_d(\R) \times L_d^2(\R)$ is given by $(\cdot, \cdot )$.
We also use the complexifications of these spaces denoted with the sub-index $\C$ (as well as their inner products and norms). The dual pairing we extend in a bilinear way. Hence we have the relation 
\begin{equation*}
\langle g,f \rangle = (\mathbf{g},\overline{\mathbf{f}}), \quad \mathbf{f},\mathbf{g} \in L_d^2(\R)_{\C},
\end{equation*}
where the overline denotes the complex conjugation.
\subsection{White Noise Spaces}
We consider on $S_d' (\R)$ the $\sigma$-algebra $\cC_{\sigma}(S_d' (\R))$ generated by the cylinder sets $\{ \omega \in S_d' (\R) | \langle \xi_1, \omega \rangle \in F_1, \dots ,\langle \xi_n, \omega \rangle \in F_n\} $, $\xi_i \in S_d(\R)$, $ F_i \in \cB(\R),\, 1\leq i \leq n,\, n\in \N$, where $\cB(\R)$ denotes the Borel $\sigma$-algebra on $\R$.\\
\noindent The canonical Gaussian measure $\mu$ on $C_{\sigma}(S_d'(\R))$ is given via its characteristic function
\begin{eqnarray*}
\int_{S_d' (\R)} \exp(i \langle {\bf f}, \boldsymbol{\omega} \rangle ) d\mu(\boldsymbol{\omega}) = \exp(- \tfrac{1}{2} \| {\bf f}\|^2 ), \;\;\; {\bf f} \in S_d(\R),
\end{eqnarray*}
\noindent by the theorem of Bochner and Minlos, see e.g.~\cite{Mi63}, \cite[Chap.~2 Theo.~1.~11]{BK95}. The space $(S_d'(\R),\cC_{\sigma}(S_d'(\R)), \mu)$ is the ba\-sic probability space in our setup.
The cen\-tral Gaussian spa\-ces in our frame\-work are the Hil\-bert spaces $(L^2):= L^2(S_d'(\R),$ $\cC_{\sigma}(S_d' (\R)),\mu)$ of complex-valued square in\-te\-grable func\-tions w.r.t.~the Gaussian measure $\mu$.\\
Within this formalism a representation of a d-dimensional Brownian motion is given by 
\begin{equation}\label{BrownianMotion}
{\bf B}_t ({\boldsymbol \omega}) :=(B_t(\omega_1), \dots, B_t(\omega_d)):= ( \langle  \1_{[0,t)},\omega_1 \rangle, \dots  \langle  \1_{[0,t)},\omega_d \rangle),\end{equation}
with ${\boldsymbol \omega}=(\omega_1,\dots, \omega_d) \in S'_d(\R),\quad t \geq 0,$
in the sense of an $(L^2)$-limit. Here $\1_A$ denotes the indicator function of a set $A$. 

\subsection{The Hida triple}

Let us now consider the Hilbert space $(L^2)$ and the corresponding Gel'fand triple
\begin{equation*}
(S) \subset (L^2) \subset (S)'.
\end{equation*}
Here $(S)$ denotes the space of Hida test functions and $(S)'$ the space of Hida distributions. In the following we denote the dual pairing between elements of $(S)$ and $(S)'$ by $\langle \! \langle \cdot , \cdot \rangle \!\rangle$. 
Instead of reproducing the construction of $(S)'$ here we give its characterization in terms of the $T$-transform.\\
\begin{definition}
We define the $T$-transform of $\Phi \in (S)'$ by
\begin{equation*}
T\Phi({\bf f}) := \langle\!\langle  \exp(i \langle {\bf f}, \cdot \rangle),\Phi \rangle\!\rangle, \quad  {\bf f}:= ({ f_1}, \dots ,{ f_d }) \in S_{d}(\R).
\end{equation*}
\end{definition}

\begin{remark}
\begin{itemize}
\item[(i)] Since $\exp(i \langle {\bf f},\cdot \rangle) \in (S)$ for all ${\bf f} \in S_d(\R)$, the $T$-transform of a Hida distribution is well-defined.
\item[(ii)] For ${\bf f} = 0$ the above expression yields $\langle\!\langle \Phi, 1 \rangle\!\rangle$, therefore $T\Phi(0)$ is called the generalized expectation of $\Phi \in (S)'$.
\end{itemize}
\end{remark}

\noindent In order to characterize the space $(S)'$ by the $T$-transform we need the following definition.

\begin{definition}
A mapping $F:S_{d}(\R) \to \C$ is called a {\emph U-functional} if it satisfies the following conditions:
\begin{itemize}
\item[U1.] For all ${\bf{f, g}} \in S_{d}(\R)$ the mapping $\R \ni \lambda \mapsto F(\lambda {\bf f} +{\bf g} ) \in \C$ has an analytic continuation to $\lambda \in \C$ ({\bf{ray analyticity}}).
\item[U2.] There exist constants $0<C,D<\infty$ and a $p \in \N_0$ such that 
\begin{equation*}
|F(z{\bf f})|\leq C\exp(D|z|^2 \|{\bf f} \|_p^2), 
\end{equation*}
for all $z \in \C$ and ${\bf f} \in S_{d}(\R)$ ({\bf{growth condition}}).
\end{itemize}
\end{definition}

\noindent This is the basis of the following characterization theorem. For the proof we refer to \cite{PS91,Kon80,HKPS93,KLPSW96}.

\begin{theorem}\label{charthm}
A mapping $F:S_{d}(\R) \to \C$ is the $T$-transform of an element in $(S)'$ if and only if it is a U-functional.
\end{theorem}
Theorem \ref{charthm} enables us to discuss convergence of sequences of Hida distributions by considering the corresponding $T$-transforms, i.e.~ by considering convergence on the level of U-functionals. The following corollary is proved in \cite{PS91,HKPS93,KLPSW96}.

\begin{corollary}\label{seqcor}
Let $(\Phi_n)_{n\in \N}$ denote a sequence in $(S)'$ such that:
\begin{itemize}
\item[(i)] For all ${\bf f} \in S_{d}(\R)$, $((T\Phi_n)({\bf f}))_{n\in \N}$ is a Cauchy sequence in $\C$.
\item[(ii)] There exist constants $0<C,D<\infty$ such that for some $p \in \N_0$ one has 
\begin{equation*}
|(T\Phi_n)(z{\bf f })|\leq C\exp(D|z|^2\|{\bf f}\|_p^2)
\end{equation*}
for all ${\bf f} \in S_{d}(\R),\, z \in \C$, $n \in \N$.
\end{itemize}
Then $(\Phi_n)_{n\in \N}$ converges strongly in $(S)'$ to a unique Hida distribution.
\end{corollary}

\begin{example}[Vector valued white noise]
\noindent Let $\,{\bf{B}}(t)$, $t\geq 0$, be the $d$-di\-men\-sional Brow\-nian motion as in \eqref{BrownianMotion}. 
Consider $$\frac{{\bf{B}}(t+h,\boldsymbol{\omega}) - {\bf{B}}(t,\boldsymbol{\omega})}{h} = (\langle \frac{\1_{[t,t+h)}}{h} , \omega_1 \rangle , \dots (\langle \frac{\1_{[t,t+h)}}{h} , \omega_d \rangle),\quad h>0.$$ 
Then in the sense of Corollary \ref{seqcor} it exists
\begin{eqnarray*}
\langle {\boldsymbol\delta_t}, {\boldsymbol \omega} \rangle := (\langle \delta_t,\omega_1 \rangle, \dots ,\langle \delta_t,\omega_d \rangle):= \lim_{h\searrow 0} \frac{{\bf{B}}(t+h,\boldsymbol{\omega}) - {\bf{B}}(t,\boldsymbol{\omega})}{h}.
\end{eqnarray*}
Of course for the left derivative we get the same limit. Hence it is natural to call the generalized process $\langle {\boldsymbol\delta_t}, {\boldsymbol \omega} \rangle$, $t\geq0$ in $(S)'$ vector valued white noise. One also uses the notation ${\boldsymbol \omega}(t) =\langle{\boldsymbol\delta_t}, {\boldsymbol \omega} \rangle$, $t\geq 0$. 
\end{example}

Another useful corollary of Theorem \ref{charthm} concerns integration of a family of generalized functions, see \cite{PS91,HKPS93,KLPSW96}.

\begin{corollary}\label{intcor}
Let $(\Lambda, \cA, \nu)$ be a measure space and $\Lambda \ni\lambda \mapsto \Phi(\lambda) \in (S)'$ a mapping. We assume that its $T$--transform $T \Phi$ satisfies the following conditions:
\begin{enumerate}
\item[(i)] The mapping $\Lambda \ni \lambda \mapsto T(\Phi(\lambda))({\bf f})\in \C$ is measurable for all ${\bf f} \in S_d(\R)$.
\item[(ii)] There exists a $p \in \N_0$ and functions $D \in L^{\infty}(\Lambda, \nu)$ and $C \in L^1(\Lambda, \nu)$ such that 
\begin{equation*}
   \abs{T(\Phi(\lambda))(z{\bf f})} \leq C(\lambda)\exp(D(\lambda) \abs{z}^2 \norm{{\bf f}}^2), 
\end{equation*}
for a.e.~$ \lambda \in \Lambda$ and for all ${\bf f} \in S_d(\R)$, $z\in \C$.
\end{enumerate}
Then, in the sense of Bochner integration in $H_{-q} \subset (S)'$ for a suitable $q\in \N_0$, the integral of the family of Hida distributions is itself a Hida distribution, i.e.~$\!\displaystyle \int_{\Lambda} \Phi(\lambda) \, d\nu(\lambda) \in (S)'$ and the $T$--transform interchanges with integration, i.e.~
\begin{equation*}
   T\left( \int_{\Lambda} \Phi(\lambda)  d\nu(\lambda) \right)(\mathbf{f}) =
   	\int_{\Lambda} T(\Phi(\lambda))(\mathbf{f}) \, d\nu(\lambda), \quad \mathbf{f} \in S_d(\R).
\end{equation*}
\end{corollary}

Based on the above theorem, we introduce the following Hida distribution.
\begin{definition}
\label{D:Donsker} 
We define Donsker's delta at $x \in \R$ corresponding to $0 \neq {\boldsymbol\eta} \in L_{d}^2(\R)$ by
\begin{equation*}
   \delta_0(\langle {\boldsymbol\eta},\cdot \rangle-x) := 
   	\frac{1}{2\pi} \int_{\R} \exp(i \lambda (\langle {\boldsymbol\eta},\cdot \rangle -x)) \, d \lambda
\end{equation*}
in the sense of Bochner integration, see e.g.~\cite{HKPS93,LLSW94,W95}. Its $T$--transform in ${\bf f} \in S_d(\R)$ is given by
\begin{equation*}
   T(\delta_0(\langle  {\boldsymbol\eta},\cdot \rangle-x)({\bf f}) 
   	= \frac{1}{\sqrt{2\pi \langle {\boldsymbol\eta}, {\boldsymbol\eta}\rangle}} \exp\left( -\frac{1}{2\langle {\boldsymbol\eta},{\boldsymbol\eta} \rangle}(i\langle {\boldsymbol\eta},{\bf f} \rangle - x)^2 -\frac{1}{2}\langle {\bf f},{\bf f}\rangle \right), \, \, \mathbf{f} \in S_d(\R).
\end{equation*}
\end{definition}

\subsection{Generalized Gauss Kernels}
Here we review a special class of Hida distributions which are defined by their $T$-transform, see e.g.~\cite{HS83},\cite{HKPS93},\cite{GS98a}. Proofs and more details for can be found in \cite{BG11}. Let $\mathcal{B}$ be the set of all continuous bilinear mappings $B:S_{d}(\R) \times S_{d}(\R) \to \C$. Then the functions
\begin{equation*}
S_d(\R)\ni \mathbf{f} \mapsto \exp\left(-\frac{1}{2} B({\bf f},{\bf f})\right) \in \C
\end{equation*}
for all $B\in \mathcal{B}$ are U-functionals. Therefore, by using the characterization of Hida distributions in Theorem \ref{charthm},
the inverse T-transform of these functions 
\begin{equation*}
\Phi_B:=T^{-1} \exp\left(-\frac{1}{2} B\right)
\end{equation*}
are elements of $(S)'$.

\begin{definition}\label{GGK}
The set of {\bf{generalized Gauss kernels}} is defined by
\begin{equation*}
GGK:= \{ \Phi_B,\; B\in \mathcal{B} \}.
\end{equation*}
\end{definition}

\begin{example}{\cite{GS98a}} \label{Grotex} We consider a symmetric trace class operator $\mathbf{K}$ on $L^2_{d}(\R)$ such that $-\frac{1}{2}<\mathbf{K}\leq 0$, then
\begin{align*}
\int_{S'_{d}(\R)} \exp\left(- \langle \omega,\mathbf{K} \omega\rangle \right) \, d\mu(\boldsymbol{\omega}) 
= \left( \det(\mathbf{Id +2K})\right)^{-\frac{1}{2}} < \infty.
\end{align*}
For the definition of $\langle \cdot,\mathbf{K} \cdot \rangle$ see the remark below.
Here $\mathbf{Id}$ denotes the identity operator on the Hilbert space $L^2_{d}(\R)$, and $\det(\mathbf{A})$ of a symmetric trace class operator $\mathbf{A}$ on $L^2_{d}(\R)$ denotes the infinite product of its eigenvalues, if it exists. In the present situation we have $\det(\mathbf{Id +2K})\neq 0$.
There\-fore we obtain that the exponential $g= \exp(-\frac{1}{2} \langle \cdot,\mathbf{K} \cdot \rangle)$ is square-integrable and its T-transform is given by 
\begin{equation*}
Tg({\bf f}) = \left( \det(\mathbf{Id+K}) \right)^{-\frac{1}{2}} \exp\left(-\frac{1}{2} ({\bf f}, \mathbf{(Id+K)^{-1}} {\bf f})\right), \quad {\bf f} \in S_{d}(\R).
\end{equation*}
Therefore $\left( \det(\mathbf{Id+K}) \right)^{\frac{1}{2}}g$ is a generalized Gauss kernel.
\end{example}

\begin{remark}
\begin{itemize}
\item[i)]\label{traceL2} Since a trace class operator is compact, see e.g.~\cite{RS75a}, we have that $\mathbf{K}$ in the above example is diagonalizable, i.e.~
\begin{equation*}
\mathbf{K}\mathbf{f} = \sum_{k=1}^{\infty} k_n (\mathbf{f},\mathbf{e}_n)\mathbf{e}_n, \quad \mathbf{f} \in L_d^2(\R),
\end{equation*}
where $(\mathbf{e}_n)_{n\in \N}$ denotes an eigenbasis of the corresponding eigenvalues $(k_n)_{n\in \N}$ with $k_n \in (-\frac{1}{2}, 0 ]$, for all $n \in \N$. Since $K$ is compact, we have that $\lim\limits_{n\to \infty} k_n =0$ and since $\mathbf{K}$ is trace class we also have $\sum_{n=1}^{\infty} (\mathbf{e}_n, -\mathbf{K} \mathbf{e}_n)< \infty$. We define for ${\boldsymbol \omega }\in S_d'(\R)$
\begin{eqnarray*} 
- \langle {\boldsymbol \omega }, \mathbf{K} {\boldsymbol \omega } \rangle := \lim_{N \to \infty} \sum_{n=1}^N \langle \mathbf{e}_n, {\boldsymbol \omega }\rangle (-k_n)\langle \mathbf{e}_n,{\boldsymbol \omega } \rangle. 
\end{eqnarray*}
Then as a limit of measurable functions ${\boldsymbol \omega } \mapsto -\langle {\boldsymbol \omega }, \mathbf{K} {\boldsymbol \omega } \rangle$  is measurable and hence 
\begin{eqnarray*} 
\int\limits_{S_d'(\R)} \exp(-  \langle {\boldsymbol \omega }, \mathbf{K} {\boldsymbol \omega }\rangle ) \, d\mu({\boldsymbol \omega }) \in [0, \infty].
\end{eqnarray*}
The explicit formula for the $T$-transform and expectation then follow by a straightforward calculation with help of the above limit procedure. 
\item[ii)] In the following, if we apply operators or bilinear forms defined on $L^2_d(\R)$ to generalized functions from $S'_d(\R)$, we are always having in mind the interpretation as in \ref{traceL2}.
\end{itemize}
\end{remark}
\begin{definition}\label{D:Nexp}\cite{BG11}$\;$
Let $\mathbf{K}: L^2_d(\R)_{\C} \to L^2_d(\R)_{\C}$ be linear and continuous such that:
\begin{itemize}
\item[(i)] $\mathbf{Id+K}$ is injective. 
\item[(ii)] There exists $p \in \N_0$ such that $(\mathbf{Id+K})(L^2_{d}(\R)_{\C}) \subset H_{p,\C}$ is dense.
\item[(iii)] There exist $q \in\N_0$ such that $\mathbf{(Id+K)^{-1}} :H_{p,\C} \to H_{-q,\C}$ is continuous with $p$ as in (ii).
\end{itemize}
Then we define the normalized exponential
\begin{equation}\label{Nexp}
{\rm{Nexp}}(- \frac{1}{2} \langle \cdot ,\mathbf{K} \cdot \rangle)
\end{equation}
by
\begin{equation*}
T({\rm{Nexp}}(- \frac{1}{2} \langle \cdot ,\mathbf{K} \cdot \rangle))({\bf f}) := \exp(-\frac{1}{2} \langle {\bf f}, \mathbf{(Id+K)^{-1}} {\bf f} \rangle),\quad {\bf f} \in S_d(\R).
\end{equation*}
\end{definition}

\begin{remark}
The "normalization" of the exponential in the above definition can be regarded as a division of a divergent factor. In an informal way one can write
\begin{multline*}
T({\rm{Nexp}}(- \frac{1}{2} \langle \cdot ,\mathbf{K} \cdot \rangle))({\mathbf f})=\frac{T(\exp(- \frac{1}{2} \langle \cdot ,\mathbf{K} \cdot \rangle))(\mathbf{f})}{T(\exp(- \frac{1}{2} \langle \cdot ,\mathbf{K} \cdot \rangle))(0)}\\
=\frac{T(\exp(- \frac{1}{2} \langle \cdot ,\mathbf{K} \cdot \rangle))(\mathbf{f})}{\sqrt{\det(\mathbf{Id+K})}} , \quad {\bf f} \in S_d(\R), 
\end{multline*}
i.e.~ if the determinant in the Example \ref{Grotex} above is not defined, we can still define the normalized exponential by the T-transform without the diverging prefactor. The assumptions in the above definition then guarantee the existence of the generalized Gauss kernel in \eqref{Nexp}.
\end{remark}

\begin{example}\label{pointprod}
	For sufficiently "nice" operators $\mathbf{K}$ and $\mathbf{L}$ on $L^2_{d}(\R)_{\C}$ we can define the product 
			\begin{equation*}
				{\rm{Nexp}}\big( - \frac{1}{2} \langle \cdot,\mathbf{K} \cdot \rangle  \big) \cdot \exp\big(-\frac{1}{2} \langle \cdot,\mathbf{L}\cdot \rangle \big)
			\end{equation*}
	of two square-integrable functions. Its $T$-transform is then given by 
			\begin{multline*}
				T\Big({\rm{Nexp}}( - \frac{1}{2} \langle \cdot,\mathbf{K} \cdot\rangle ) \cdot \exp( - \frac{1}{2} \langle \cdot,\mathbf{L} \cdot\rangle )\Big)({\bf f})\\
				=\sqrt{\frac{1}{\det(\mathbf{Id+L(Id+K)^{-1}})}}
				\exp(-\frac{1}{2} \langle {\bf f}, \mathbf{(Id+K+L)^{-1}} {\bf f} \rangle ),\quad {\bf f} \in S_{d}(\R),
			\end{multline*}	
	in the case the right hand side indeed is a U-functional.		
\end{example}

\begin{definition}\label{prodnexp}
Let $\mathbf{K}: L^2_{d}(\R)_{\C} \to L^2_{d}(\R)_{\C}$ be as in Definition \ref{D:Nexp}, i.e.~$${\rm{Nexp}}(- \frac{1}{2} \langle \cdot ,\mathbf{K} \cdot \rangle)$$ exists. Furthermore let $\mathbf{L}: L^2_d(\R)_{\C} \to L^2_d(\R)_{\C}$ be trace class. Then we define 
$$
{\rm{Nexp}}( - \frac{1}{2} \langle \cdot,\mathbf{K} \cdot\rangle ) \cdot \exp( - \frac{1}{2} \langle \cdot,\mathbf{L} \cdot\rangle )$$ via its $T$-transform, whenever 
\begin{multline*}
				T\Big({\rm{Nexp}}( - \frac{1}{2} \langle \cdot,\mathbf{K} \cdot\rangle ) \cdot \exp( - \frac{1}{2} \langle \cdot,\mathbf{L} \cdot\rangle )\Big)({\bf f})\\
				=\sqrt{\frac{1}{\det(\mathbf{Id+L(Id+K)^{-1}})}}
				\exp(-\frac{1}{2} \langle {\bf f}, \mathbf{(Id+K+L)^{-1}} {\bf f} \rangle ),\quad {\bf f} \in S_{d}(\R),
			\end{multline*}	
is a U-functional.
\end{definition}

In the case $\mathbf{g} \in S_d(\R)$, $c\in\C$ the product between the Hida distribution $\Phi$ and the Hida test function $\exp(i \langle \mathbf{g},. \rangle + c)$ can be defined because $(S)$ is a continuous algebra under point-wise multiplication. The next definition is an extension of this product.

\begin{definition}\label{linexp}
The point-wise product of a Hida distribution $\Phi \in (S)'$ with an exponential of a linear term, i.e.~
\begin{equation*}
\Phi \cdot \exp(i \langle {\bf g}, \cdot \rangle  +c), \quad {\bf g} \in L^2_{d}(\R)_{\C}, \, c \in \C,
\end{equation*}
is defined by 
\begin{equation*}
T(\Phi \cdot \exp(i\langle  {\bf g}, \cdot \rangle  + c))({\bf f}):= T\Phi({\bf f}+{\bf g})\exp(c),\quad {\bf f} \in S_d(\R),  
\end{equation*}
if $T\Phi$ has a continuous extension to $L^2_d(\R)_{\C}$ and the term on the right-hand side is a U-functional in ${\bf f} \in S_d(\R)$.
\end{definition}

\begin{definition}\label{donsker}
Let $D \subset \R$ with $0 \in \overline{D}$. Under the assumption that $T\Phi$ has a continuous extension to $L^2_d(\R)_{\C}$, ${\boldsymbol\eta}\in L^2_d(\R)_{\C}$, $y \in \R$, $\lambda \in \gamma_{\alpha}:=\{\exp(-i\alpha)s|\, s \in \R\}$ and that the integrand 
\begin{equation*}
\gamma_{\alpha} \ni \lambda \mapsto \exp(-i\lambda y)T\Phi({\bf f}+\lambda {\boldsymbol\eta}) \in \C
\end{equation*}
fulfills the conditions of Corollary \ref{intcor} for all $\alpha \in D$. Then one can define the product 
\begin{equation*}
\Phi \cdot \delta_0(\langle {\boldsymbol\eta}, \cdot \rangle-y),
\end{equation*}
by
\begin{equation*}
T(\Phi \cdot \delta_0(\langle {\boldsymbol\eta}, \cdot \rangle-y))({\bf f})
:= \lim_{\alpha \to 0} \int_{\gamma_{\alpha}} \exp(-i \lambda y) T\Phi({\bf f}+\lambda {\boldsymbol\eta}) \, d \lambda.
\end{equation*}
Of course under the assumption that the right-hand side converges in the sense of Corollary \ref{seqcor}, see e.g.~\cite{GS98a}.
\end{definition}

This definition is motivated by the definition of Donsker's delta, see Definition \ref{D:Donsker}.

\begin{lemma}{\cite{BG11}}\label{thelemma}
Let  $\mathbf{L}$ be a $d\times d$ block operator matrix on $L^2_{d}(\R)_{\C}$ acting component-wise such that all entries are bounded operators on $L^2(\R)_{\C}$.
Let $\mathbf{K}$ be a d $\times d$ block operator matrix on $L^2_{d}(\R)_{\C}$, such that $\mathbf{Id+K}$ and $\mathbf{N}=\mathbf{Id}+\mathbf{K}+\mathbf{L}$ are bounded with bounded inverse. Furthermore assume that $\det(\mathbf{Id}+\mathbf{L}(\mathbf{Id}+\mathbf{K})^{-1})$ exists and is different from zero (this is e.g.~the case if $\mathbf{L}$ is trace class and -1 in the resolvent set of $\mathbf{L}(\mathbf{Id}+\mathbf{K})^{-1}$).
Let $M_{\mathbf{N}^{-1}}$ be the matrix given by an orthogonal system $({\boldsymbol\eta}_k)_{k=1,\dots J}$ of non--zero functions from $L^2_d(\R)$, $J\in \N$, under the bilinear form $\left( \cdot ,\mathbf{N}^{-1} \cdot \right)$, i.e.~ $(M_{\mathbf{N}^{-1}})_{i,j} = \left( {\boldsymbol\eta}_i ,\mathbf{N}^{-1} {\boldsymbol\eta}_j \right)$, $1\leq i,j \leq J$.
Under the assumption that either 
\begin{eqnarray*}
\Re(M_{\mathbf{N}^{-1}}) >0 \quad \text{ or }\quad \Re(M_{\mathbf{N}^{-1}})=0 \,\text{ and } \,\Im(M_{\mathbf{N}^{-1}}) \neq 0,
\end{eqnarray*} 
where $M_{\mathbf{N}^{-1}}=\Re(M_{\mathbf{N}^{-1}}) + i \Im(M_{\mathbf{N}^{-1}})$ with real matrices $\Re(M_{\mathbf{N}^{-1}})$ and $\Im(M_{\mathbf{N}^{-1}})$, \\
then
\begin{equation*}
\Phi_{\mathbf{K},\mathbf{L}}:={\rm Nexp}\big(-\frac{1}{2} \langle \cdot, \mathbf{K} \cdot \rangle \big) \cdot \exp\big(-\frac{1}{2} \langle \cdot, \mathbf{L} \cdot \rangle \big) \cdot \exp(i \langle \cdot, {\bf g} \rangle)
\cdot \prod_{i=1}^J \delta_0 (\langle \cdot, {\boldsymbol\eta}_k \rangle-y_k),
\end{equation*}
for ${\bf g} \in L^2_{d}(\R,\C),\, t>0,\, y_k \in \R,\, k =1\dots,J$, exists as a Hida distribution. \\
Moreover for ${\bf f} \in S_d(\R)$
\begin{multline}\label{magicformula}
T\Phi_{\mathbf{K},\mathbf{L}}({\bf f})=\frac{1}{\sqrt{(2\pi)^J  \det((M_{\mathbf{N}^{-1}}))}}
\sqrt{\frac{1}{\det(\mathbf{Id}+\mathbf{L}(\mathbf{Id}+\mathbf{K})^{-1})}}\\ 
\times \exp\bigg(-\frac{1}{2} \big(({\bf f}+{\bf g}), \mathbf{N}^{-1} ({\bf f}+{\bf g})\big) \bigg)
\exp\bigg(-\frac{1}{2} (u,(M_{\mathbf{N}^{-1}})^{-1} u)\bigg),
\end{multline}
where
\begin{equation*}
u= \left( \big(iy_1 +({\boldsymbol\eta}_1,\mathbf{N}^{-1}({\bf f}+{\bf g})) \big), \dots, \big(iy_J +({\boldsymbol\eta}_J,\mathbf{N}^{-1}({\bf f}+{\bf g})) \big) \right).
\end{equation*}
\end{lemma}

\section{Charged particle in a constant magnetic field}
\label{magnetic}
In this subsection we want to calculate the transition amplitude for the movement of a charged particle in a constant magnetic field. Investigations of this system for the Feynman integrand had been done in white noise analysis in \cite{G96}, \cite{J10} and \cite{BGJ13}.  Here we consider the motion in the plane orthogonal to the direction of the magnetic field. Note that the propagator in three dimension can be obtained by multiplying the expression with the free motion propagator along the axis of the magnetic field vector. 
As in \ref{cpaction} we have 
\begin{multline}
S(q,p,t) = \int_0^{t} {\bf p} \dot{\vec{{\bf x}}} - \frac{1}{2m} \bigg({\bf p} -\frac{q}{c} B_z \left(\begin{array}{c} -x_2\\x_1 \end{array}\right)  \bigg)^2 \; d\tau\\ \nonumber
= \int_0^{t} {\bf p} \dot{\vec{{\bf x}}} - \frac{1}{2m}(p_1^2+p_2^2) +\frac{q}{mc} B_z (x_1 p_2 -x_2 p_1) -\frac{q^2 B_z^2}{2mc^2}(x_1^2+x_2^2) d\tau.
\end{multline}
For simplicity we set $x_{1,0}=x_{2,0}=0, t_0=0$ and $m=\hbar=1$. Then with $k = \frac{q B_z}{mc}$ we obtain
\begin{eqnarray*}
I_{CP}={\rm N}\exp\big(-\frac{1}{2} \langle \cdot, K \cdot \rangle \big) \cdot \exp\big(-\frac{1}{2} \langle \cdot, L \cdot \rangle\big) \cdot \delta\big(\langle  {\bf 1}_{[0,t)}, \cdot_{x1} \rangle -y_1 \big)\delta\big(\langle {\bf 1}_{[0,t)}, \cdot_{x2} \rangle -y_2 \big),
\end{eqnarray*}
as in \ref{anpsfeycp}.
For the kinetic energy part and the local simulation of the flat measure we obtain the operator matrix 
\begin{eqnarray*} 
K=\left(
\begin{array}{l l l l }
-{\bf 1}_{[0,t)} & -i{\bf 1}_{[0,t)}& 0 & 0 \\
-i{\bf 1}_{[0,t)}& -{\bf 1}_{[0,t)}+i{\bf 1}_{[0,t)} & 0 & 0\\
0 & 0& -{\bf 1}_{[0,t)} & -i{\bf 1}_{[0,t)}\\
0 &0& -i{\bf 1}_{[0,t)}& -{\bf 1}_{[0,t)}+i{\bf 1}_{[0,t)}
\end{array}
\right).
\end{eqnarray*}
In addition we have to model the potential. We use an ansatz where we have an upper triangular block matrix, i.e.
\begin{eqnarray*}
L=\left(
\begin{array}{l l l l}
ik^2 A & 0 & 0 & -2ikB^* \\
0 & 0 & 2ikB & 0 \\
0 & 0 & ik^2 A & 0\\
0&0&0&0 \\
\end{array}
\right),
\end{eqnarray*} 
with 
\begin{eqnarray*}
\R \ni s \mapsto Af(s) =& \1_{[0,t)}(s)\int_s^t \int_0^r f(\tau) \, d\tau \, dr\text{ and }\\
\R \ni s \mapsto Bf(s) =& \1_{[0.t)}(s) \int_0^s f(r) \, dr,
\end{eqnarray*}
for $f \in L^2(\R)_{\C}$, $B^*$ denotes the dual operator of $B$ w.r.t.~the dual pairing. Note that we have $\langle Ag,f\rangle = \langle Bg,Bf\rangle$ for all $f,g \in L^2(\R)_{\C}$. 
Thus 
\begin{multline*}
Id+K+L=N\\
= \left(\begin{array}{l l l l} {\bf 1}_{[0,t)^c}& 0 & 0 & 0\\
 0&{\bf 1}_{[0,t)^c} & 0 & 0 \\
 0& 0& {\bf 1}_{[0,t)^c} & 0 \\
 0&0& 0& {\bf 1}_{[0,t)^c} 
 \end{array} \right)
 +\left(
\begin{array}{l l l l}
ik^2 A &-i{\bf 1}_{[0,t)} & 0 & -2ik B^*\\
-i{\bf 1}_{[0,t)} & i{\bf 1}_{[0,t)} & 2ik B & 0 \\
0 & 0 & ik^2 A & -i{\bf 1}_{[0,t)} \\
0& 0& -i{\bf 1}_{[0,t)}& i{\bf 1}_{[0,t)}
\end{array}
\right).
\end{multline*}

In the sequel we identify the subspace of functions from $L^2_d(\R)_{\C}$ zero on $[0,t)$ or $[0,t)^c$ with  $L^2_d([0,t))_{\C}$ or $L^2_d([0,t)^c)_{\C}$, respectively. Then we have the orthogonal decomposition:
$$L^2_d(\R)_{\C}=L^2_d([0,t))_{\C} \bot L^2_d([0,t)^c)_{\C}.$$

\begin{proposition}\label{Ntriboundedinvertible}
The operator 
\begin{multline*}
Id+K+L=N\\
= \left(\begin{array}{l l l l} {\bf 1}_{[0,t)^c}& 0 & 0 & 0\\
 0&{\bf 1}_{[0,t)^c} & 0 & 0 \\
 0& 0& {\bf 1}_{[0,t)^c} & 0 \\
 0&0& 0& {\bf 1}_{[0,t)^c} 
 \end{array} \right)
 +\left(
\begin{array}{l l l l}
ik^2 A &-i{\bf 1}_{[0,t)} & 0 & -2ik B^*\\
-i{\bf 1}_{[0,t)} & i{\bf 1}_{[0,t)} & 2ik B & 0 \\
0 & 0 & ik^2 A & -i{\bf 1}_{[0,t)} \\
0& 0& -i{\bf 1}_{[0,t)}& i{\bf 1}_{[0,t)}
\end{array}
\right)
\end{multline*}
is linear and bounded on $L^2_4(\R)_{\C}$ and has a bounded inverse.
\end{proposition}

\begin{proof}
The operator $A$ is in the trace class, moreover $B$ and $B^*$ are compact operators. Moreover $A_{\mid{L^2([0,t))_{\C}}}: L^2([0,t))\to L^2([0,t))$ is invertible. For convenience, since the operator $A$ leaves $L^2([0,t))_{\C}$ invariant and acts trivial on $L^2([0,t)^c)_{\C}$, we denote $A_{\mid L^2([0,t))_{\C}}$ by $A$, when there is no danger of confusion. Since $N$ is the identity when restricted on $L^2_4([0,t)^c)_{\C}$, we can restrict ourselves to $L^2_4([0,t))_{\C}$. We then have
\begin{multline*}
N_{\mid L_4^2([0,t))_{\C}}=\left(
\begin{array}{l l l l}
ik^2 A &-i{\bf 1}_{[0,t)} & 0 & -2ik B^*\\
-i{\bf 1}_{[0,t)} & i{\bf 1}_{[0,t)} & 2ik B & 0 \\
0 & 0 & ik^2 A & -i{\bf 1}_{[0,t)} \\
0& 0& -i{\bf 1}_{[0,t)}& i{\bf 1}_{[0,t)}
\end{array}
\right)\\=i\left(
\begin{array}{l l l l}
k^2 A &-{\bf 1}_{[0,t)} & 0 & -2k B^*\\
-{\bf 1}_{[0,t)} & {\bf 1}_{[0,t)} & 2k B & 0 \\
0 & 0 & k^2 A & -{\bf 1}_{[0,t)} \\
0& 0& -{\bf 1}_{[0,t)}& {\bf 1}_{[0,t)}
\end{array}
\right),
\end{multline*}
which is of the form
$$R=\left(
\begin{array}{l l}
M&P\\0&M
\end{array}
\right),$$
with a bounded invertible matrix $M$. The inverse of such a matrix is given by
$$R^{-1}=\left(
\begin{array}{l l}
M^{-1}& -M^{-1} P M^{-1} \\0&M^{-1}
\end{array}
\right).$$
Indeed
\begin{multline*}
\left(
\begin{array}{l l}
M&P\\0&M
\end{array}
\right)\left(
\begin{array}{l l}
M^{-1}& -M^{-1} P M^{-1} \\0&M^{-1}
\end{array}
\right) \\
=\left(
\begin{array}{l l}
M M^{-1}& -M M^{-1} P M^{-1} + PM^{-1} \\0&MM^{-1}
\end{array}
\right)
=
\left(
\begin{array}{l l}
Id& 0\\0&Id
\end{array}
\right).\end{multline*}
Now in our case 
$$
M^{-1}=\left(\begin{array}{l l}
(kA- \1_{[0,t)})^{-1} &  (kA- \1_{[0,t)})^{-1} \\
(kA- \1_{[0,t)})^{-1}& kA(kA- \1_{[0,t)})^{-1}
\end{array} \right).$$
Thus
\begin{align*}
M^{-1}PM^{-1}&=\left(\begin{array}{l l}
(k^2A- \1_{[0,t)})^{-1} &  (k^2A- \1_{[0,t)})^{-1} \\[.3 cm]
(k^2A- \1_{[0,t)})^{-1}& k^2A(k^2A- \1_{[0,t)})^{-1}
\end{array} \right)\\[.3 cm]
&\times\left(\begin{array}{l l}
0 &  -2kB^* \\[.3 cm]
2kB& 0
\end{array} \right)
\left(\begin{array}{l l}
(k^2A- \1_{[0,t)})^{-1} &  (k^2A- \1_{[0,t)})^{-1} \\[.3 cm]
(k^2A- \1_{[0,t)})^{-1}& k^2A(k^2A- \1_{[0,t)})^{-1}
\end{array} \right)\\[.3 cm]
&=\left(\begin{array}{l l}
(k^2A- \1_{[0,t)})^{-1} &  (k^2A- \1_{[0,t)})^{-1} \\[.3 cm]
(k^2A- \1_{[0,t)})^{-1}& k^2A(k^2A- \1_{[0,t)})^{-1}
\end{array} \right)\\[.3 cm]
&\times\left(\begin{array}{l l}
-2kB^*(k^2A- \1_{[0,t)})^{-1} & -2k^3B^*A(kA- \1_{[0,t)})^{-1} \\[.3 cm]
2kB(k^2A- \1_{[0,t)})^{-1}& 2kB(k^2A- \1_{[0,t)})^{-1}
\end{array} \right).
\end{align*}
Hence we have that $N_{|L_4^2([0,t))_{\C}}$ is bounded invertible. Thus $N$ is invertible, since $N_{|L_4^2([0,t)^c)_{\C}}=Id$. Moreover the inverse yields: 
\begin{multline}\label{Ninvcp}
N^{-1}=\left(\begin{array}{l l l l} \1_{[0,t)^c}& 0 & 0 & 0\\[0.3 cm]
 0&{\1}_{[0,t)^c} & 0 & 0 \\[0.3 cm]
 0& 0& \1_{[0,t)^c} & 0 \\[0.3 cm]
 0&0& 0& \1_{[0,t)^c} 
 \end{array} \right)\\
 +\frac{1}{i}\left([0.3 cm]
\begin{array}{c c c c}
\1_{[0,t)} &  \1_{[0,t)} &2k(k^2A- \1_{[0,t)})^{-1} (B-B^*) & 2 (k^2A- \1_{[0,t)})^{-1} (kB-k^3B^*A)\\[0.3 cm]
\1_{[0,t)}& k^2A
&-2 (k^2A- \1_{[0,t)})^{-1} (kB^*-k^3AB)& 2k^3 (k^2A- \1_{[0,t)})^{-1} (AB-B^*A)\\[0.3 cm]
0 & 0 & \1_{[0,t)} &  \1_{[0,t)}\\[0.3 cm]
0& 0&\1_{[0,t)} & k^2A
\end{array}
\right)\\[0.3 cm]
\times
\left(
\begin{array}{l l l l}
(k^2A- \1_{[0,t)})^{-1} &  0 &0 & 0\\[0.3 cm]
0& (k^2A- \1_{[0,t)})^{-1}
&0&0\\[0.3 cm]
0 & 0 & (k^2A- \1_{[0,t)})^{-1} &  0 \\[0.3 cm]
0& 0& 0& (k^2A- \1_{[0,t)})^{-1}
\end{array}
\right).
\end{multline}
\end{proof}
Next we calculate the matrix $M_{N^{-1}}$ as in Lemma \ref{thelemma}.

\begin{proposition}
For $N$ as is Proposition \ref{Ntriboundedinvertible} we have for ${\boldsymbol\eta_1}= (\1_{[0,t)},0, 0, 0)$ and ${\boldsymbol\eta_3}= (0, 0, \1_{[0,t)},0)$, with $0<t<\infty$, $t\neq \frac{(2n-1)\pi}{2k},\quad n\in \N$, that
$$
M_{N^{-1}}=\left(\begin{array}{l l}
\langle {\boldsymbol\eta_1},N^{-1} {\boldsymbol\eta_1}\rangle & \langle \boldsymbol {\boldsymbol\eta_1}, N^{-1} {\boldsymbol\eta_3}\rangle\\
\langle {\boldsymbol\eta_3}, N^{-1} {\boldsymbol\eta_1} \rangle &\langle {\boldsymbol\eta_3},N^{-1} {\boldsymbol\eta_3}\rangle \end{array}\right)
=i\left(\begin{array}{l l} 
\frac{\tan\!\left(k\, t\right)\,}{k}
 & 0
\\
0&
\frac{\tan\!\left(k\, t\right)\,}{k}
 \end{array}\right),$$
moreover the assumptions as in Lemma \ref{thelemma} are fulfilled.
\end{proposition}
\begin{proof}
We have ${\boldsymbol\eta_1}=(\1_{[0,t)},0, 0, 0)$ and ${\boldsymbol\eta_3}= (0, 0, \1_{[0,t)},0)$. Hence, by Proposition \ref{Ntriboundedinvertible}, we just have to consider the restriction of $N^{-1}$ to $ L^2_4([0,t))_{\C}$.\\
Instead of calculating the inverse directly we find a preimage of ${\boldsymbol\eta_1}$ and ${\boldsymbol\eta_3}$, respectively, under the operator $N$. 
We have 
$$\left(
\begin{array}{l l l l}
k^2 A &-\1_{[0,t)} & 0 & -2k B^*\\
-\1_{[0,t)} & \1_{[0,t)} & 2k B & 0 \\
0 & 0 & k^2 A & -\1_{[0,t)} \\
0& 0& -\1_{[0,t)}& \1_{[0,t)}
\end{array}\right) \left(\begin{array}{l} f_1\\f_2\\f_3\\f_4 \end{array}\right)=-i {\boldsymbol\eta_k}\quad k=1,3.
$$
We can transfer this to a system of differential equations, note that the function on the right-hand-side is almost surely constant. We obtain
\begin{align}
(I)&\quad -k^2f_1 +2k f_4'=f_2''\\
(II)&\quad f_1'-2kf_3=f_2'\\
(III)&\quad -k^2 f_3=f_4''\\
(IV)&\quad f_3=f_4.
\end{align}
Taking into account that $f_3=f_4$ and deriving equation $(II)$ and setting it equal to $(I)$ we obtain.
\begin{align}
(I)&\quad -k^2f_1 +2k f_4'=f_1''-2kf_3'\\
(II)&\quad f_1'-2kf_3=f_2'\\
(III)&\quad -k^2 f_3=f_3''\\
(IV)&\quad f_3=f_4.
\end{align}
Then $(I)$ can be written as
$$f_1'' = -k^2f_1 +4kf_3'.$$
To obtain now the preimages we have to take the boundary conditions into account. 
We have by the definition of $B$ and $B^*$ and taking into account that ${\boldsymbol\eta_k}_2={\boldsymbol\eta_k}_4=0$ the following boundary conditions:
$$ f_1(0) =f_2(0)$$
$$ f_2'(0)= f_3(0)$$
$$f_4'(0)=f_3'(0)=0.$$
The additionally two boundary conditions are obtained by inserting ${\boldsymbol\eta_k}$.
For ${\boldsymbol\eta_1}$ we have
$$f_2(t)=i,$$
$$f_4(t)=f_3(t)=0.$$ 
For ${\boldsymbol\eta_3}$ we have
$$f_2(t)=0,$$
$$f_4(t)=f_3(t)=i.$$ 
We solved this system of differential equations with the dsolve-routine in MATLAB and obtained
$$N{\bf f}=(\1_{[0,t)},0,0,0)$$ with 
$${\bf f} =\left(\begin{array}{l}
\frac{\cos\!\left(k\, s\right)\, \mathrm{i}}{\cos\!\left(k\, t\right)} \\
\frac{\cos\!\left(k\, s\right)\, \mathrm{i}}{\cos\!\left(k\, t\right)}\\
0\\
0
\end{array}\right),$$
which is well-defined due to our restrictions on $t$ and 
$$N{\bf h}=(0,0,\1_{[0,t)},0)$$ with
$${\bf h} =\left(\begin{array}{l}
\frac{2\, \sin\!\left(k\, s\right)\, \mathrm{i} + 2\, k\, s\, \cos\!\left(k\, s\right)\, \mathrm{i} - 2\, k\, t\, \cos\!\left(k\, s\right)\, \mathrm{i}}{\cos\!\left(k\, t\right)} \\
\frac{2\, k\, \cos\!\left(k\, s\right)\, \left(s - t\right)\, \mathrm{i}}{\cos\!\left(k\, t\right)}
\\
\frac{\cos\!\left(k\, s\right)\, \mathrm{i}}{\cos\!\left(k\, t\right)}
\\
\frac{\cos\!\left(k\, s\right)\, \mathrm{i}}{\cos\!\left(k\, t\right)}
\end{array}\right),$$
which again exists due to our restrictions on $t$.\\[0.5 cm] 

Then we have
$$
M_{N^{-1}}=\left(\begin{array}{l l}
\langle {\boldsymbol\eta_1}, {\bf f}\rangle & \langle {\boldsymbol\eta_1}, {\bf h}\rangle\\
\langle {\boldsymbol\eta_3},  {\bf f}\rangle &\langle {\boldsymbol\eta_3},  {\bf h}\rangle \end{array}\right)
=i\left(\begin{array}{l l} 
\frac{\tan\!\left(k\, t\right)\,}{k}
 & 0
\\
0&
\frac{\tan\!\left(k\, t\right)\,}{k}
 \end{array}\right).$$
\end{proof}

Now, to have all ingredients for the integrand, we calculate the determinant of
\begin{multline*}
(Id +L(I+K)^{-1}) =\left(\begin{array}{l l l l} \1_{[0,t)^c} & 0 & 0 & 0 \\
0&\1_{[0,t)^c} & 0 & 0 \\
0&0&\1_{[0,t)^c} &  0 \\
0&0&0&\1_{[0,t)^c} \end{array} \right)\\
+
\left(\begin{array}{l l l l} \1_{[0,t)}-k^2A & -k^2A &2kB^* &0\\
0&\1_{[0,t)} & -2kB & -2kB \\
0&0&\1_{[0,t)}-k^2A &  -k^2A \\
0&0&0&\1_{[0,t)}\end{array} \right)
\end{multline*}


\begin{proposition}\label{eigen_structure_of_Id_L_K}
Let $0<t<\infty$, $t\neq \frac{(2n-1)\pi}{2k},\quad n\in \N$. Then the eigenvalues of ${Id +L(Id+K)^{-1}}_{\mid{L_4^2([0,t))_{\C}}}$ are $v_{0}=1$ and
\begin{equation*}
v_n = 1-k^2  \bigg(\frac{t}{(n-\frac{1}{2})\pi}\bigg)^2, \quad n =1,2,3\cdots.
\end{equation*}
The algebraic multiplicity of each $v_n$, $n=1,2,3,\cdots$ is 2.\\\\
The eigenfunctions to $v_{0}$ have the form
$$\left(\begin{array}{l} f_1\\f_2\\f_3\\-f_3 \end{array}\right),$$
where $f_2,f_3$ are arbitrarily chosen in $L^2([0,t))_{\C}$ and $f_1$ solves the equation
$$k^2Af_1=-k^2Af_2+2kB^*f_3.$$\\
The eigenfunctions to $v_n,\ n=1,2,3,\cdots$ is the set
$$\left\{\alpha\left(\begin{array}{c} \cos(\frac{k}{\sqrt{1-v_n}}s)\\0\\0\\0 \end{array}\right) + \beta\left(\begin{array}{c} \frac{2k}{1-v_n} s \cos(\frac{k}{\sqrt{1-v_n}}s)\\ \frac{2}{\sqrt{1-v_n}}\sin(\frac{k}{\sqrt{1-v_n}}s)\\ \cos(\frac{k}{\sqrt{1-v_n}}s)\\0\end{array}\right) \quad \Bigg|\ \alpha,\beta\in\C\right\}.$$
\\
Hence 
\begin{multline*}
\det(Id +L(I+K)^{-1}) = det(Id-k^2 A)^2 = \left(\prod_{n=1}^{\infty} 1-k^2(\frac{t^2}{(n-\halb)^2\pi^2}\right)^2 = \cos(kt)^2.
\end{multline*}
\end{proposition}
\begin{proof}
First note that we need the restriction on $t$ for the well-definiteness of $(Id+K)^{-1}$. We consider the equation
\begin{equation}\label{eqeigenfunc}
(Id +L(Id+K)^{-1})\ \left(\begin{array}{l} f_1\\f_2\\f_3\\f_4 \end{array}\right)=\lambda\left(\begin{array}{l} f_1\\f_2\\f_3\\f_4 \end{array}\right),\ \lambda\in\C,\ f_1,\ f_2,\ f_3,\ f_4\ \in L^2([0,t))_{\C}\end{equation}
Note that $Id +L(Id+K)^{-1}=N(Id+K)^{-1}$ is invertible because N is invertible. Hence 0 is not an eigenvalue of $Id +L(Id+K)^{-1}$. Furthermore $L(Id+K)^{-1}$ is a Hilbert-Schmidt operator, where each eigenvalue different from zero has a finite algebraic multiplicity. Consequently each eigenvalue of $Id +L(Id+K)^{-1}$ which is different from 1 has a finite algebraic multiplicity. Since we do not know whether $L(Id+K)^{-1}$ is a trace class operator, it is not sure that the determinant of $Id +L(Id+K)^{-1}$ has a finite value. We first calculate the eigenfunctions.\\
\underline{Case 1: $\lambda=1$.}\\
The forth component in \eqref{eqeigenfunc} for $\lambda=1$ is fulfilled for all $f_4$, hence we choose $f_4$ arbitrarily. The third component gives $(1-k^2A)f_3-k^2Af_4=f_3$, hence $f_3=-f4$. Using this result in the second component we get $f_2=f_2$, hence $f_2$ can also be chosen arbitrarily. The first equation leads to
$$k^2Af_1=-k^2Af_2+2kB^*f_3$$
as condition for $f_1$. The eigenvectors of $Id +L(Id+K)^{-1}_{|{L^2([0,t))_{\C}}}$ corresponding to the eigenvalue 1 are given by $\left(\begin{array}{c} f_1\\f_2\\ \frac{k}{2}B(f_1+f_2)\\ -\frac{k}{2}B(f_1+f_2) \end{array}\right)$,
with $f_1,\ f_2\in L^2([0,t))_{\C}$.
\\
\underline{Case 2: $\lambda\neq 1$.}\\
The 4th component in \eqref{eqeigenfunc} implies $f_4=0$.
Assuming $f_3=0$ implies $f_2=0$ and $f_1$ is an eigenvector of $1-k^2A$ corresponding to $\lambda$, hence $\lambda=v_n$ for some $n>0$.\\
Next we assume $f_3\neq 0$. Then $f_3$ is an eigenvector of $1-k^2A$ corresponding to $\lambda$, hence $\lambda=v_n$ for some $n>0$. Furthermore $f_2=\frac{2k}{1-\lambda}Bf_3$ by the second component and $$f_1\in\left\{(1-\lambda-k^2A)^{-1}\left((\frac{2k^3}{1-\lambda} AB-2kB^*)f_3\right)\, \Big| \, f_3 \in \text{span}\left\{\cos(\frac{k}{\sqrt{1-v_n}}s)\right\} \right\}$$ by the first component. Note that the set $$\left\{(1-\lambda-k^2A)^{-1}\left((\frac{2k^3}{1-\lambda} AB-2kB^*)f_3\right)\, \Big| \,f_3 \in \text{span}\left\{\cos(\frac{k}{\sqrt{1-v_n}}s)\right\} \right\}$$ is not empty, since we have  
\begin{multline*}
\left\{(1-\lambda-k^2 A)^{-1}\left( \frac{2k^3}{1-\lambda}AB-2kB^* \right)f_3\right\}\\
= \frac{2k}{1-\lambda}\cdot \left(s\ f_3 \right)+\ker((1-\lambda)-k^2A).\end{multline*}
So the eigenspace of $Id +L(Id+K)^{-1}\mid_{L_4^2([0,t))_{\C}}$ corresponding to $v_n$ is the set
$$\left\{\alpha\left(\begin{array}{l} f_1\\0\\0\\0 \end{array}\right) + \beta\left(\begin{array}{c} \frac{2k}{1-v_n}\cdot s\ f_3\\ \frac{2k}{1-v_n}Bf_3\\f_3\\0\end{array}\right) \Bigg|\ f_1, f_3 \text{ eigenvectors of }1 - k^2A\ \text{ to } v_n,\ \alpha,\beta\in\C\right\}.$$
Thus we obtain
$$\det(Id +L(Id+K)^{-1}_{|{L_4^2([0,t))_{\C}}})=\prod\limits_{n=1}^{\infty}\left(1-k^2  \bigg(\frac{t}{(n-\frac{1}{2})\pi}\bigg)^2\right)\left(1-k^2  \bigg(\frac{t}{(n-\frac{1}{2})\pi}\bigg)^2\right).$$
Since $L$ leaves $L_4^2([0,t))_{\C}$ invariant and acts trivial on the complement, we have that the eigenvalues of $Id +L(Id+K)^{-1}_{|{L_4^2([0,t)^c)_{\C}}}$ are one. Thus the determinant of the operator on this subspace equals one. Therefore we have:
\begin{multline*}
\det(Id +L(I+K)^{-1}) = \det(Id +L(Id+K)^{-1}_{|{L_4^2([0,t))_{\C}}})\\
= \det(Id-k^2 A)^2 = \left(\prod_{n=1}^{\infty} 1-k^2(\frac{t^2}{(n-\halb)^2\pi^2}\right)^2 = \cos(kt)^2.
\end{multline*}
\end{proof}


%
Thus we can state the following theorem.
\begin{theorem}\label{magnetictheorem}
Let $y_1, y_2\in \R$, $0<t<\infty$, $t\neq \frac{(2n-1)\pi}{2k},\quad n\in \N$, then the Feynman integrand $I_{cp}$ for the charged particle in a constant magnetic field in phase space  exists as a Hida distribution and for ${\boldsymbol \xi}\in S_4(\R)$ its generating functional is given by \begin{multline}\label{genfuncp}
T(I_{cp})({\boldsymbol \xi})
= \left(\frac{kt}{2\pi i \sin(kt)}\right) \exp\bigg(-\frac{1}{2} \big\langle{\boldsymbol \xi}, N^{-1} {\boldsymbol \xi}\big\rangle \bigg)\\
\times\exp\bigg(\frac{1}{2i} (u^T \left(\begin{array}{l l} 
\frac{kt}{\tan\!\left(k\, t\right)\,}
 & 0
\\
0&
\frac{kt}{\tan\!\left(k\, t\right)\,}
 \end{array}\right) u)\bigg),\\
\text{ with } u= \left( \begin{array}{l}
iy_1 +\halb \langle{\boldsymbol\eta}_1,N^{-1}{\boldsymbol \xi}\rangle \big)+\halb \langle N^{-1}{\boldsymbol\eta}_1,{\boldsymbol \xi}\rangle \\ \dots\\ iy_2 +\halb\langle{\boldsymbol\eta}_3,N^{-1}{\boldsymbol \xi}\rangle+\halb \langle N^{-1}{\boldsymbol\eta}_3,{\boldsymbol \xi}\rangle  \end{array}\right).
\end{multline} 
Moreover its generalized expectation
\begin{equation*}
\mathbb{E}(I_{cp})=T(I_{cp})(0)=\left(\frac{k}{2\pi i \sin(kt)}\right) \exp\left( i \frac{k}{2\tan(k t)} (y_1^2+y_2^2)\right)
\end{equation*}
is the Greens function to the Schrö\-dinger equation for the charged particle in a constant magnetic field, compare e.g.~with \cite{KL85}.   
\end{theorem}

\end{document}